\documentclass[sort&compress]{article}

\usepackage{fullpage}
\usepackage{amsmath,amssymb,amsfonts}
\usepackage{graphicx}
\usepackage[numbers]{natbib}
\usepackage{cleveref}
\usepackage{caption}
\usepackage{subcaption}
\usepackage{authblk}

\newtheorem{thm}{Theorem}
\newtheorem{lemma}[thm]{Lemma}
\newtheorem{cor}[thm]{Corollary}

\newtheorem{defn}{Definition}
\newtheorem{remark}{Remark}

\newtheorem{ex}{Example}
\newenvironment{proof}{\noindent {\sc Proof:}}{$\Box$ \medskip}

\newcommand{\beginsupplement}{%
        \setcounter{table}{0}
        \renewcommand{\thetable}{S\arabic{table}}%
        \setcounter{figure}{0}
        \setcounter{section}{0}
        \renewcommand{\thefigure}{S\arabic{figure}}%
        \renewcommand{\thesection}{S\arabic{section}}
     }

\title{A Formal Definition of Scale-dependent Complexity and the Multi-scale Law of Requisite Variety}
\date{}

\author[1,2]{Alexander F. Siegenfeld}
\author[2]{Yaneer Bar-Yam}
\affil[1]{Department of Physics, Massachusetts Institute of Technology, Cambridge, MA}
\affil[2]{New England Complex Systems Institute, Cambridge, MA}

\begin{document}

\maketitle

\begin{abstract}
Ashby's law of requisite variety allows a comparison of systems with their environments, providing a necessary (but not sufficient) condition for system efficacy: a system must possess at least as much complexity as any set of environmental behaviors that require distinct responses from the system.  However, to account for the dependence of a system's complexity on the level of detail---or scale---of its description, a multi-scale generalization of Ashby's law is needed.  We define a class of complexity profiles (complexity as a function of scale) that is the first, to our knowledge, to exhibit a multi-scale law of requisite variety.  This formalism provides a characterization of multi-scale complexity and generalizes the law of requisite variety's single constraint on system behaviors to a class of multi-scale constraints.  We show that these complexity profiles satisfy a sum rule, which reflects a tradeoff between smaller- and larger-scale degrees of freedom, and we extend our results to subdivided systems and systems with a continuum of components. 
\end{abstract}

\section{Introduction}
Defining complexity in general terms has been a persistent challenge in the study of complex systems~\cite{gell2002complexity, adami2002complexity, crutchfield2012between, lineweaver2013complexity}.  Scale-dependent complexity---which has been formalized for general collections of random variables~\cite{textbook, ay2006unifying, allen2017multiscale} and time series in particular~\cite{binder2001multiscale, ahmed2011multivariate, wu2013time, humeau2015multiscale} and used in contexts such as chaos~\cite{zunino2012distinguishing, he2018multivariate}, biological signals~\cite{costa2002multiscale, costa2003multiscale, costa2005multiscale, catarino2011atypical, liang2012automatic}, traffic patterns~\cite{wang2013multiscale, yin2016multivariate}, financial data~\cite{martina2011multiscale,alvarez2012multiscale,yin2014weighted}, Gaussian processes~\cite{faes2017multiscale}, and fluid dynamics~\cite{li2014permutation,murayama2018characterization}---offers a promising path: 
rather than attempt to describe the complexity of a system with a single number, it recognizes that the average length of description necessary to specify a system's state (i.e. the system's complexity) depends on the level of detail or scale of that description (see \cref{fig:behaviors}). 

\begin{figure}
\begin{center}
\includegraphics[width=.8\textwidth]{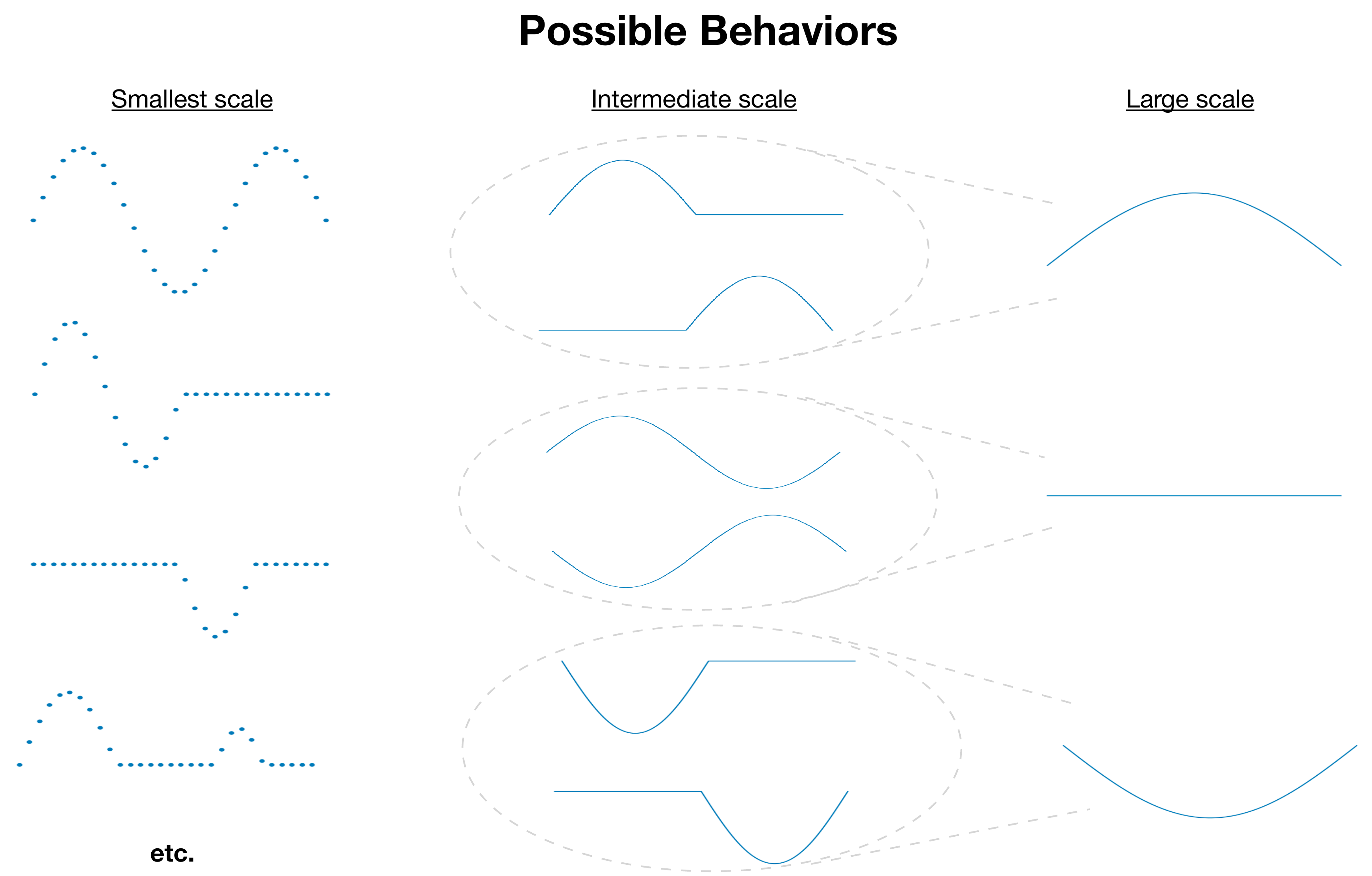}
\caption{The average length of description necessary to specify a system's state depends on its number of possible distinct behaviors.  Which behaviors are distinct depends on the scale/level of detail at which the system is being described.  At the smallest scale, the system depicted has many possible states, but many distinguishable smaller-scale states can all correspond to a single larger-scale state.}
\label{fig:behaviors}
\end{center}
\end{figure}

But, in terms of real-world consequences, what does it mean if one system possesses more complexity than another?  Ashby's law of requisite variety~\cite{ashby1961introduction} provides one possible answer: in order to be effective, the complexity of a system must equal or exceed that of its environment.  Here, a system's environment must be defined to include only that set of behaviors that require a distinct response from the system (see \cref{fig:ashby}).

However, Ashby's law does not account for the multi-scale nature of complexity.  For instance, two individuals who lacked the ability to cooperate may have enough small-scale complexity to independently move various objects but would lack sufficient complexity at the scale necessary to move a couch.  Thus, while the two individuals would possess sufficient complexity (i.e. sufficient variety of behaviors), it would not be at the right scale (i.e. at the level of description that includes only behaviors that involve coordination between the two individuals).  This example is not in violation of Ashby's law, as Ashby's law merely provides a necessary rather than sufficient condition for system efficacy.  Nonetheless, this example and others (see below) motivate us to seek a stronger necessary condition for system efficacy that takes into account not only the complexity  but also the scale of system behaviors.  We thus propose a \textit{multi-scale law of requisite variety}: in order to be effective, the complexity of a system must equal or exceed that of its environment \textit{at all scales}.   Our goal is to provide a definition of scale-dependent complexity---a class of complexity profiles---that satisfies this multi-scale version law of requisite variety.  For a pedagogical introduction to complexity profiles and the multi-scale law of requisite variety, please see ref.~\cite{siegenfeld2020}. 

\begin{figure}
\begin{center}
\includegraphics[width=.5\textwidth]{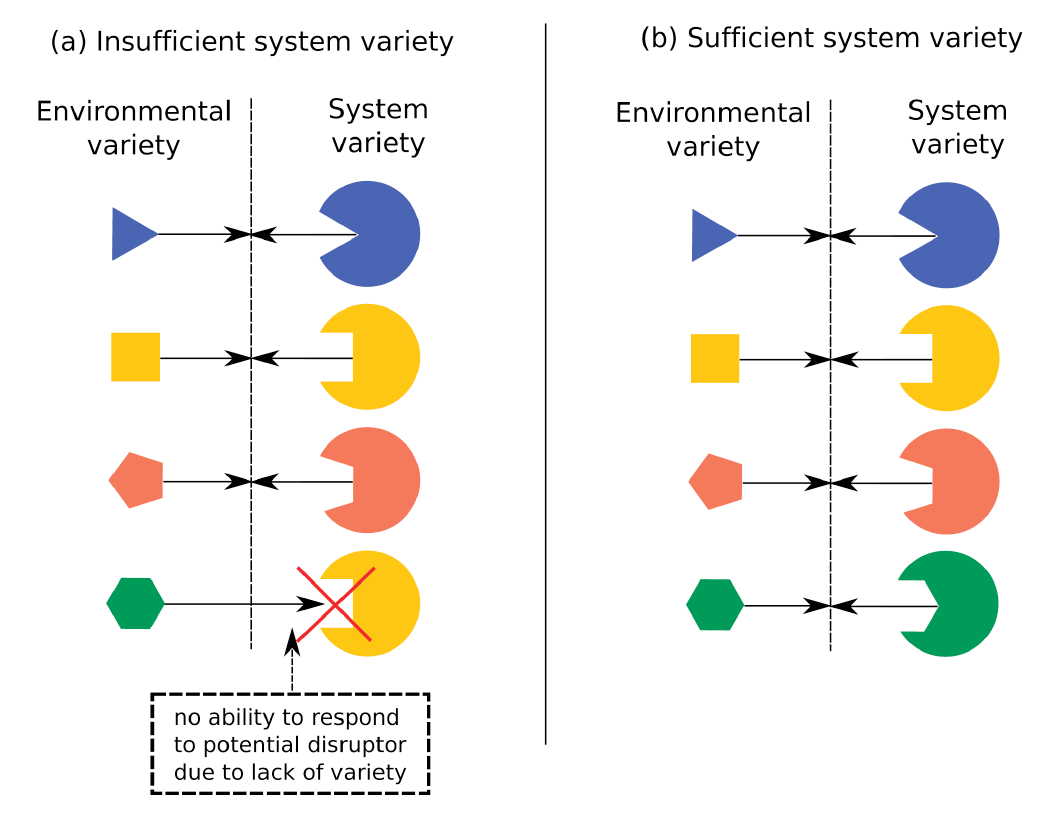}
\caption{An illustration of Ashby's law of requisite variety.  (a) Because the system has fewer states (i.e. lower complexity) than its environment, it is impossible for the system to have a distinct response to each of the four environmental states.  (b) Here, the system is able to have a distinct response to each environmental state; a necessary (but not sufficient) condition for this matching is that the system's complexity equals or exceeds its environment's.  Image source: ref.~\cite{norman2019}.}
\label{fig:ashby}
\end{center}
\end{figure}

For an existing formal definition of a complexity profile~\cite{original,allen2017multiscale}, it has been proven that a multi-scale law of requisite variety applies for systems and environments that are \textit{block-independent}, i.e. systems/environments whose components can be partitioned into mutually independent blocks such that components within the same block have identical behavior~\cite{multivariety}.  However, the law of requisite variety does not apply to this complexity profile more generally.  For instance, adding additional components to the system (without changing the existing components) can actually reduce this complexity profile at larger scales due to the possibility of negative interaction information for more than two variables~\cite{stronge}; thus, a system that is capable of effectively interacting with its environment could nonetheless end up with less complexity than its environment at larger scales.  Given the desirability and usefulness of a complexity profile satisfying Ashby's law at each scale (a property that has been implicitly used in many analyses, such as management~\cite{rosenkranz2011variety, mckelvey2012organisations, gorod2019}, military defense~\cite{ryan2008bears,galway2011tasking,norman2019}, governance~\cite{huang2022democracies}, multi-agent coordination~\cite{alexiou2011understanding, mahmoodi2018complexity, salas2021social}, and evolutionary dynamics~\cite{derosa2007combined, gershenson2011sigma, stacey2015multiscale}), we therefore seek a formal definition of the complexity profile that reflects this property.   (We will show that for block-independent systems, the formalism introduced here can be reduced to the definition of complexity profiles discussed above.)  

The one other constraint that we desire for a complexity profile is a sum rule, i.e. that the area under the complexity profile does not depend on interdependencies between components but rather only the individual components' behaviors.   Such a constraint reflects the tradeoff between complexities at various scales: in order for a system to have complexity at larger scales, its components must be correlated, which constrains the fine-scale configurations of the system (and thus its smaller-scale complexity)~\cite{siegenfeld2020}.  Without a sum rule or some other similar constraint, the multi-scale law of requisite variety would be no more than many copies of the single-scale version of Ashby's law---one for each scale---with no structure relating the various scales to each other. 

In order to define a complexity profile for which the multi-scale law of requisite variety holds, we first have to define what it means for a multi-component system to effectively match its environment (\cref{sec:gen}).  Then, in \cref{sec:prof}, we formally define what constitutes a complexity profile and what criterion must be satisfied for it to capture the multi-scale law of requisite variety and the tradeoff between complexities at various scales.  In \cref{sec:class}, we define a class of complexity profiles that satisfy such criteria and examine some of its properties.   

Such a class does not provide a single complexity profile; rather, a complexity profile is assigned for each way of partitioning the system.   Choosing a method of partitioning the system is analogous to choosing a coordinate system onto which the multi-scale complexity of the system can be projected (breaking the permutation symmetry corresponding to the relabeling of system components). The fact that the profile depends on the partitioning method reflects the fact that there is no single way to coarse-grain a system, although some coarse-graining choices are more useful/better reflect system structure than others.  However, the efficacy of a system in a particular environment is independent of the choice of coordinate system; thus, regardless of which partitioning scheme is chosen, an effective system will have at least as much complexity as the environment at all scales.  This formalism therefore gives an entire class of constraints that must be satisfied: as long as the system and its environment are partitioned/coarse-grained in the same way, the system's complexity matching/exceeding its environment's at all scales provides a necessary condition for system efficacy.   

\section{Generalizing Ashby's Law to Multiple Components}
\label{sec:gen}
Ashby's law claims that to effectively regulate an environment, the system must have a degree of freedom or behavior for each distinct environmental behavior.  In other words, there cannot be two environmental states for a given system state.  It then follows (by the pigeon-hole principle) that the number of behaviors of the system must be greater than or equal to the number of behaviors of the environment.   

More formally, let $X$ and $Y$ be random variables or collections of random variables, let $H(X)$ denote the Shannon entropy of $X$, which is the minimum average number of bits needed to describe the state of $X$, and let $H(Y|X)$ denote the expected value of the Shannon entropy of $Y$ given the state of $X$~\cite{cover2012elements}.   Then, each system state of $X$ corresponding to no more than one environmental state of $Y$ can be written as $H(Y|X)=0$, where $H(Y|X)$ denotes the expected value of the Shannon entropy of $Y$ given knowledge of $X$,  from which it follows that $H(X)\geq H(Y)$ (i.e. the complexity of the environment can not exceed that of the system).  .
It is important to note that in this formulation, the environment $Y$ is \textit{defined} to be the set of states that require distinct behaviors of the system.  Two environmental states that do not require different system behaviors should be represented by a single state of $Y$.

In order to consider multi-scale behavior, let us describe the system $X$ as consisting of $N$ components such that $X=\{x_1,...,x_N\}$. 

\begin{defn}
\label{def:sys}
A \emph{system} $X$ of size $N=|X|$ is defined as a set of $N$ random variables.  These random variables are referred to as \emph{components} of the system.
\end{defn}

\begin{figure}
\begin{center}
\includegraphics[width=.5\textwidth]{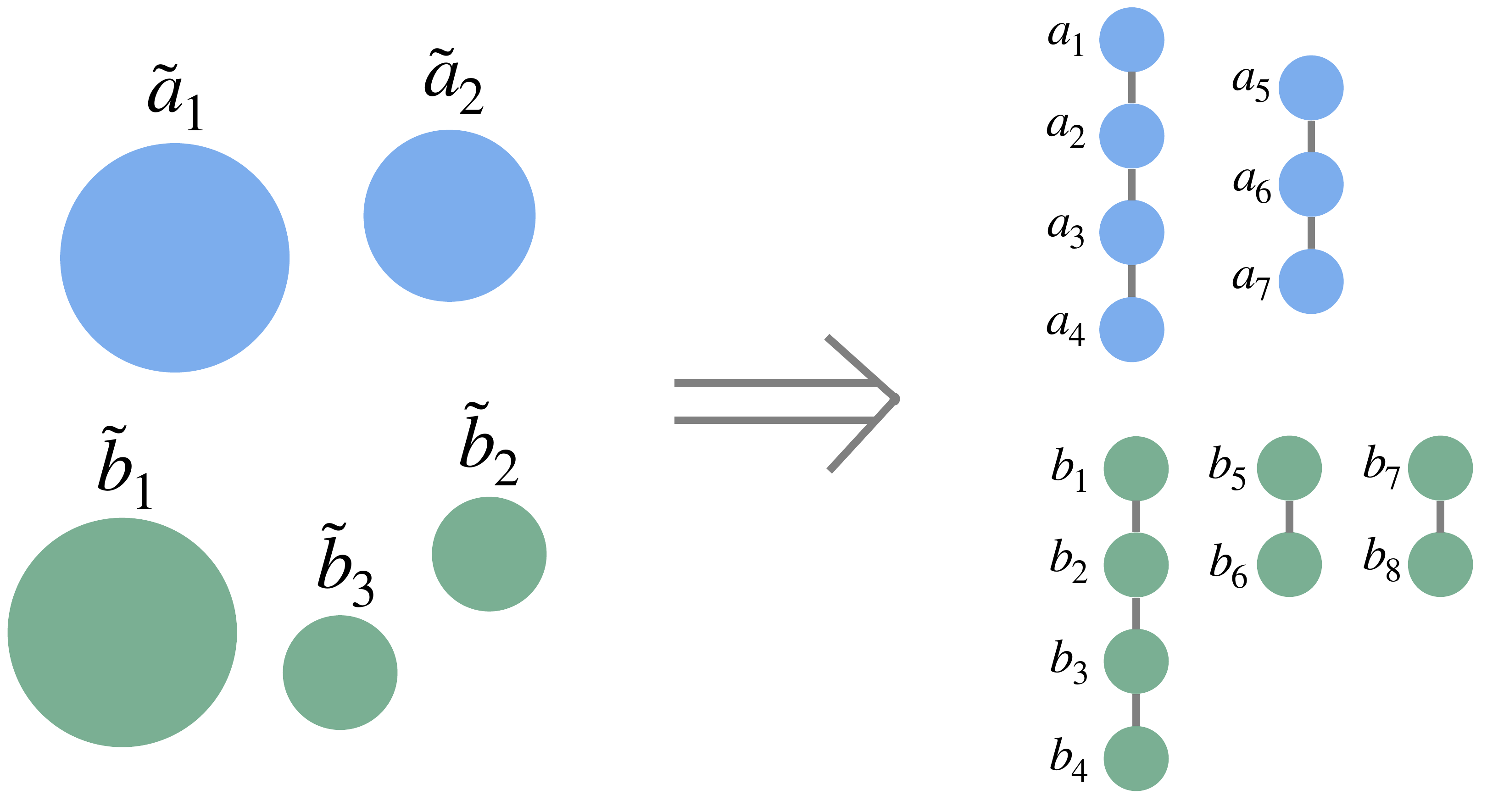}
\caption{Accounting for components of various sizes.  All systems can be described by components that are all of the same size.  For instance, two systems $\{\tilde a_1, \tilde a_2\}$ and $\{\tilde b_1, \tilde b_2, \tilde b_3\}$ that contain components of sizes $2s$, $1.5s$, and $1s$ (where the units of $s$ depend on the notion of size being used) can be reformulated in terms of components $\{a_1,...,a_7\}$ and $\{b_1,...,b_8\}$ that all have size $0.5s$, where $a_1=a_2=a_3=a_4=\tilde a_1$, $a_5=a_6=a_7=\tilde a_2$, $b_1=b_2=b_3=b_4=\tilde b_1$, $b_5=b_6=\tilde b_2$, and $b_7=b_8=\tilde b_3$.}
\label{fig:sizes}
\end{center}
\end{figure} 

\begin{remark}
In this formulation, all components of one or more systems are treated as the same size; while this condition may seem like a limitation, any system or systems can be described in this way to arbitrary precision: components of different sizes can be accounted for by defining a new set of components whose size is the greatest common factor of the sizes of the original components (irrational relative sizes---for which no greatest common factor exists---can be approximated to arbitrary precision by rational relative sizes).  If the new components are all of size $l$, each original component $\tilde x_i$ of size $l_i$ can then be replaced with $l_i/l$ new components $x_{j+1},x_{j+2}...x_{j+l_i/l}$ for which the state of one of these $l_i/l$ variables completely determines the state of all the others, i.e. $x_{j+1}=x_{j+2}=...=x_{j+l_i/l}$ (see e.g. \cref{fig:sizes}). 
\end{remark}

The assumption that the system $X$ must have at least one distinct response for each environmental state $Y$ (i.e. $H(Y|X)=0$) is generalized as follows: an ``environmental component'' $y_i$ is defined for each system component $x_i\in X$, such that each $y_i$ is a random variable representing the environmental states that require a distinct response from the system component $x_i$.  Then, for the system to effectively interact with its environment, $H(y_i|x_i)=0$ for each $i$, i.e. there cannot be two environmental component states for a given state in the corresponding system component.  (Note that this condition is necessary but not sufficient for the system to effectively interact with its environment---just because the system components can choose a different response for each environmental condition does not guarantee that the responses are appropriate.)   Letting $Y=\{y_1,...,y_N\}$, we see that $H(y_i|x_i)=0$ implies $H(Y|X)=0$ and, thus $H(y_i|x_i)=0$ is a stronger condition: not only must the system match the environment overall, but this matching must be properly organized in a specific way.  This formulation allows for constraints among the environmental components to induce constraints among the system components. 

Note that each environmental component represents the space of behaviors that are required by its corresponding system component, which as shown in \cref{fig:match}, can correspond to multiple physical parts of the environment.  Depending on how the system is connected to the environment---or depending on which constraints arising from interactions between the system and its environment are to be examined---the environmental components may need to be defined differently.  Although it may not seem so at first, any interaction between a system and its environment can be formulated as above: if we start with a more general formulation in which each system component $x_i$ interacts with environmental components $\tilde y_{i_1},\tilde y_{i_2},...,\tilde y_{i_{n_i}}$ (i.e. $\forall i,j,\ H(y_{i_j}|x_i)=0 $), which allows for each system component to interact with multiple environmental components and vice versa, then we can redefine the environmental components such that each $x_i$ is associated with the random variable $y_i\equiv(\tilde y_{i_1},\tilde y_{i_2},...,\tilde y_{i_{n_i}})$ (see \cref{fig:match} for an example).  This redefinition may result in new entanglements among the environmental components.

\begin{figure}
\begin{center}
\includegraphics[width=.5\textwidth]{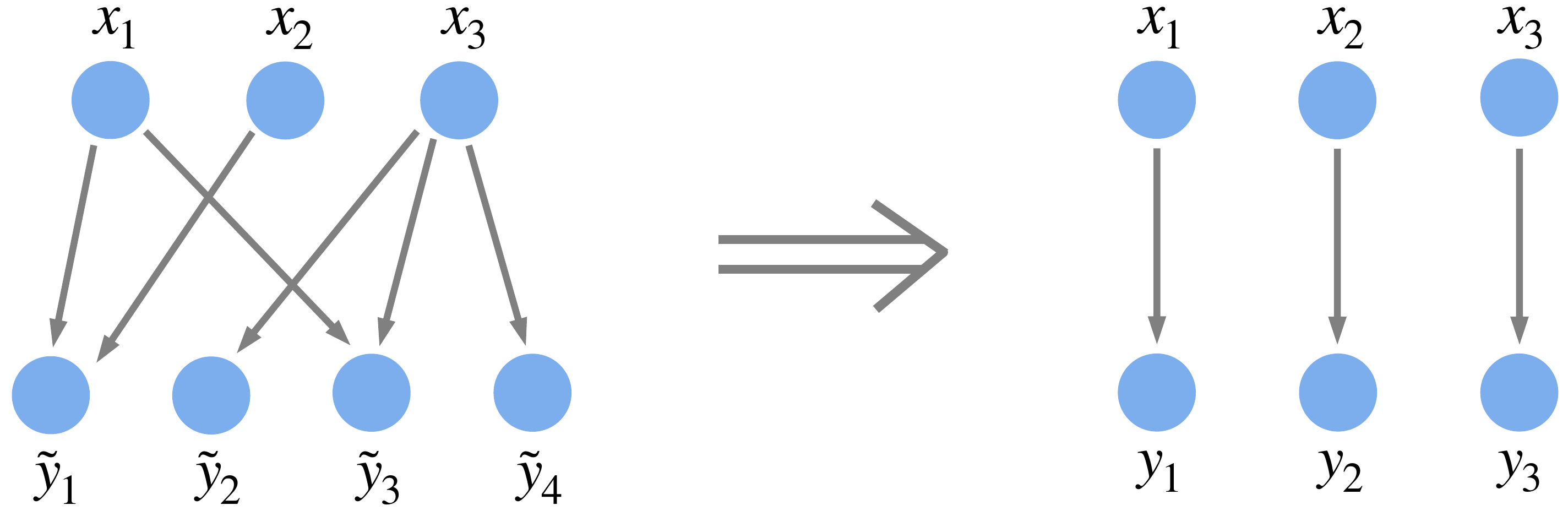}
\caption{Defining the environmental components.  Regardless of the interactions between a system and its environment, environmental components can always be defined such that they have a one-to-one relationship with system components.  For instance, suppose that for the system $\{x_1,x_2,x_3\}$ to effectively interact with its environment $\{\tilde y_1,\tilde y_2, \tilde y_3,\tilde y_4\}$, $x_1$ must have a distinct response for each possible state of $\tilde y_1$ and $\tilde y_3$ (i.e. $H(\tilde y_1|\tilde x_1)=0$ and $H(\tilde y_3|\tilde x_1)=0$), $x_2$ must have a distinct response for each possible state of $\tilde y_1$ (i.e. $H(\tilde y_1|x_2)=0$, and $x_3$ must have a distinct response for each possible state of $\tilde y_2$, $\tilde y_3$, and $\tilde y_4$ (i.e. $H(\tilde y_2|x_3)=0$,~~$H(\tilde y_3|x_3)=0$, and $H(\tilde y_4|x_3)=0$))).  If we define new environmental components $y_1\equiv(\tilde y_1, \tilde y_3)$, $y_2\equiv \tilde y_1$, and $y_3\equiv (\tilde y_2, \tilde y_3, \tilde y_4)$, we see that $x_1$ must react only to $y_1$, $x_2$ only to $y_2$, and $x_3$ only to $y_3$, since the original constraints are equivalent to the constraints $H(y_i|x_i)=0$ for $i\in\{1,2,3\}$.}
\label{fig:match}
\end{center}
\end{figure}

\begin{defn}
\label{def:env}
An \emph{environment} $(Y,f)$ for system $X$ is a system $Y$ together with a bijection $f:Y\rightarrow X$. 
\end{defn}

\begin{defn}
\label{def:match}
A system $X$ \emph{matches} its environment $(Y,f)$ iff $H(y|f(y))=0$ for all $y\in Y$.
\end{defn}

\begin{ex}
Consider a system of two thermostats, $x_1$ in room 1 and $x_2$ in room 2, that can each be either on or off.   The environment can be described by two variables $y_1$ and $y_2$ that represent whether or not room 1 or room 2, respectively, should be heated (the bijection $f$ mapping $y_1$ to $x_1$ and $y_2$ to $x_2$).  In order for the system to match the environment, it must be that $H(y_1|x_1)=H(y_2|x_2)=0$.  Thus if the need for each room to be heated is independent of that of the other room, the thermostats must be able to operate independently of one another; likewise if the two rooms' need for heat are correlated, the thermostats must also be correlated.\footnote
{Note that this formalism has nothing to say regarding causation.  It may be that the objective is to heat both rooms at the same time or not at all, in which case the thermostats themselves must be connected.  Or it may be that it just so happens that the two rooms get cold at the same time, in which case two disconnected thermostats may nonetheless exhibit correlated behavior.}
\end{ex}

\begin{ex}
Consider a system $X$ in which each $x_i\in X$ represents some aspect of policy (e.g. educational policy) being applied in region $i$ of a given country (the regions could, for instance, be towns/cities).  The environment $(Y,f)$ could be defined by the random variables $y_i$ (where $f(y_i)=x_i$), such that each $y_i$ corresponds to conditions in region $i$ that require a distinct policy in order for the region to be effectively governed.  If the $y_i$ vary independently of one another, while the $x_i$ cannot, then the system $X$ will not be able to match its environment $(Y,f)$ (e.g. an education policy that is determined entirely at the national level will not be able to effectively interact with locales if each locale has specific educational needs).  Conversely, if there are correlations among the $y_i$ that are lacking in the $x_i$, the system will also be unable to match its environment (e.g. it would be ineffective for each city to independently set its own policy with respect to international trade or with respect to regulating a national corporation that spans many cities).
\end{ex}

Note that the possible states of a system or the probabilities assigned to these states cannot be defined without specifying the environment with which the system is interacting, for the same system may behave differently in different environments.  (Alternatively, each individual environment need not be treated separately; if each individual environment is assigned a probability, this ensemble of possible individual environments can itself be treated as a single environment of the system.)  In either case, this formalism concerning a system matching its environment is purely descriptive and does not require the specification of the mechanism by which the system and the environment are related.

\begin{ex}
Returning to the example of the two thermostats, if the system (the thermostats) and the environment (the rooms) are connected so that the state of thermostat $i$ depends directly on the state of room $i$, then the thermostat states will have precisely as much correlation as the room states do.  The thermostat states will be independent random variables if and only if the room states are.  
\end{ex}

With \cref{def:match}, we have a characterization of  Ashby's law that takes into account the multi-scale structure of a system and its connection with its environment.  The goal is then to understand how properties of the environment constrain the corresponding properties of the system.  If an environment has a certain property and it is known that the system matches the environment, what must be true about the system?  For the single-scale case of Ashby's law, the system must have at least as much information as the environment.  The complexity profile, described below, generalizes this property to multiple scales.  In particular, it allows us to formulate the multi-scale law of requisite variety: in order for a system to match its environment, it must have at least as much complexity as its environment \textit{at every scale}.

\section{Defining a complexity profile}
\label{sec:prof}
The basic version of Ashby's law states that for a system $X$ to match its environment $Y$, the overall complexity of $X$ must be greater than or equal to the overall complexity of $Y$.  But, as argued in section 2.3 of ref.~\cite{siegenfeld2020}, it does not make sense to speak of complexity as a single number but rather the complexity of a system must depend on its scale. Thus, we wish to generalize the notion of a complexity profile such that the complexity of a system and its environment can be compared at multiple scales. 

\begin{defn}
\label{def:genc}
A complexity profile $C_X(n)$ of a system $X$ assigns a particular amount of information to the system at each scale $n\in\mathbb{Z}^+$.  For $n>|X|$, we define $C_X(n)=0$.  If we wish to consider each component of the system to be of size $l$, we can define a continuous version of the complexity profile (see \cref{sec:continuum} for more detail):
\begin{equation}
\label{eq:ctilde}
\tilde C_X(s)=C_X(\lceil s/l \rceil)
\end{equation}
\end{defn}

We wish for a complexity profile to have two additional properties:  it should (1) manifest the \textit{multi-scale law of requisite variety} and (2) obey the \textit{sum rule}.   Each property is defined below, with applications/examples given in sections 2.5 and 2.4 of ref.~\cite{siegenfeld2020}, respectively. 

\begin{defn}
\label{def:mono}
A complexity profile manifests the multi-scale law of requisite variety if, for any two systems $X$ and $Y$, $X$ matching $Y$ (per \cref{def:match}) implies that $C_X(n)\geq C_Y(n)$ for all $n$.  
\end{defn}

\begin{defn}
\label{def:sum}
A complexity profile obeys the sum rule if for any system $X$, 
$\sum_{n=1}^\infty C_X(n)=\sum_{x\in X}H(x)$.
\end{defn}

The multi-scale law of requisite variety is important because it allows for the interpretation that a necessary (but not sufficient) condition for a system to effectively interact with its environment must be that it has at least as much complexity as the environment at every scale.  The sum rule is important because it captures the intuition that for a system composed of components with the same individual behaviors, there is a tradeoff among the complexities of the system at various scales, since complexity at larger scales requires constraints among the system's smaller-scale degrees of freedom.

Note that examining measures of multi-scale complexity can never prove that a system matches its environment---just as in the single-scale case, a system having more complexity than its environment by no means guarantees that every system state corresponds to a single environmental state (nor that the state adopted by the system will be appropriate).  But examining multi-scale measures of information \textit{can} prove the impossibility of compatibility.  The goal then, in formulating multi-scale measures, is to create more instances in which the impossibility of compatibility can be shown.  Using this multi-scale formalism, the system must now possess more complexity than its environment \textit{at all scales}, not just more complexity than its environment overall.  For instance, an army of ants may have more fine-grained complexity than its environment but will be able to perform certain tasks (e.g. moving large objects) only with larger-scale coordination between the ants.

\section{A class of complexity profiles}
\label{sec:class}
In \cref{sec:prof}, we have defined the term \textit{complexity profile} and have described general properties that any complexity profile should have.  We now describe a specific class of complexity profiles that satisfy these properties.  This class of profiles is not the only such class and may not be the best one, but it serves as an instructive example and provides one useful way of characterizing multi-scale complexity.  

\begin{figure}
\begin{center}
\includegraphics[width=.5\textwidth]{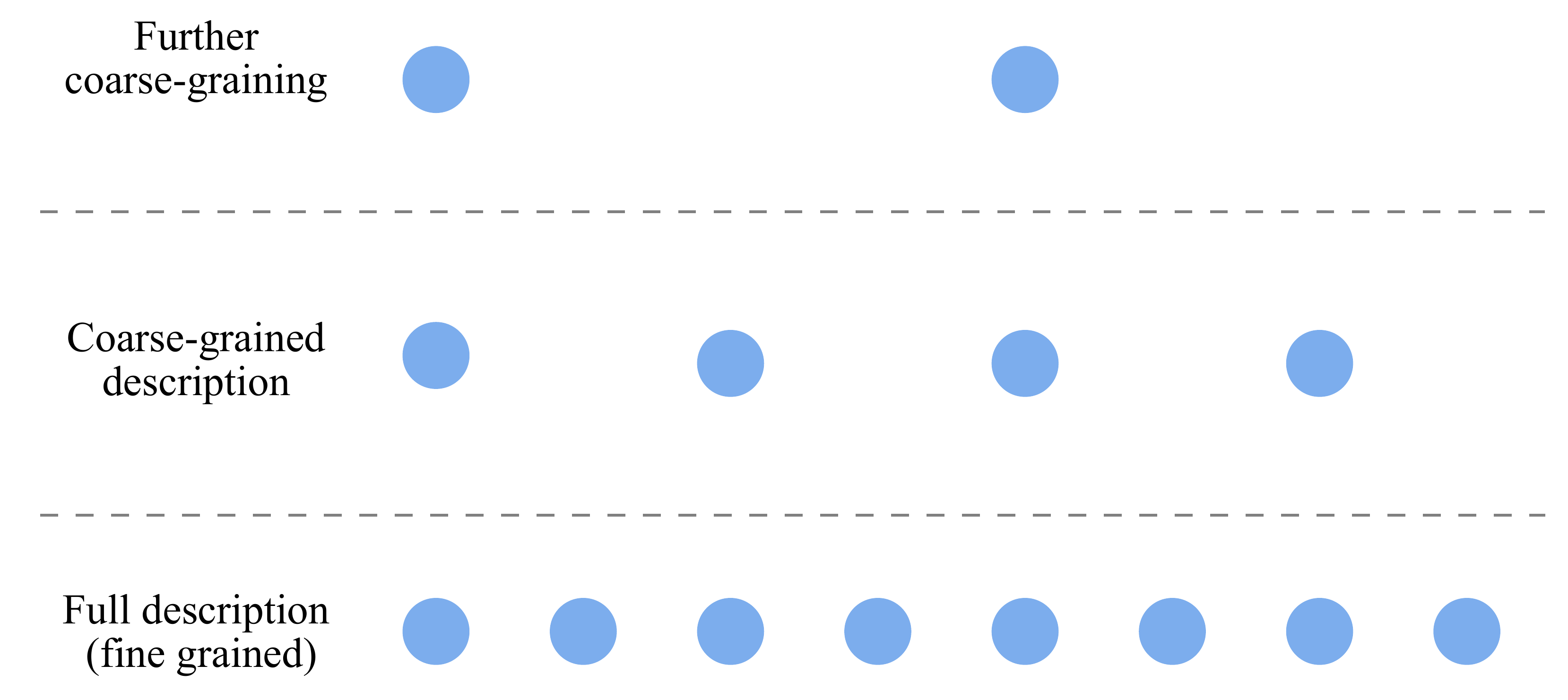}
\caption{Larger-scale/coarse-grained descriptions.  Consider a system with eight components.  The full description (scale 1) consists of all eight components.  A coarse-grained description (scale 2) might consist of every other component, which can serve as an approximation for the system as a whole.  A further coarse-grained description (scale 4) might consist of every other component of the scale-2 description.}
\label{fig:coarse-graining}
\end{center}
\end{figure} 

One way to define a large-scale or coarse-grained description of a system is to allow only a subset of the components of the system to be described.\footnote
{This coarse-graining scheme is analogous to the decimation approach for implementing the position-space renormalization group in physics.}
As a first pass, one might divide the system into $n$ equivalent disjoint subsets and then define the information in the description of the system at scale $n$ to simply be the information in one of the subsets (see e.g. \cref{fig:coarse-graining}).  However, given that the partition into $n$ equivalent subsets may not be possible (either due to heterogeneity in the components or because the system size is not divisible by $n$), this definition can be generalized by averaging over the the information in each of the $n$ subsets.  

\begin{ex}
Consider a Markov chain $(x_1,x_2,x_3,...)$ (for finite Markov chains of size $N$, simply let $x_i\equiv 0$ for $i> N$).  A set of disjoint, coarse-grained descriptions of the Markov chain at scale $n$ could be $$\{(x_1,x_{1+n},x_{1+2n},x_{1+3n},...),~(x_2,x_{2+n},x_{2+2n},x_{2+3n},...),~(x_3,x_{3+n},x_{3+2n},x_{3+3n},...),...,~(x_n,x_{2n},x_{3n},...)\}$$
Thus, the information at scale $n$ of the Markov chain could be defined as 
\begin{equation}
\frac{1}{n}\sum_{i=1}^n H(\{x_j|j\equiv i\mod n\})
\end{equation}
Note, however, that this sequence of descriptions is not nested, and so cannot be used in its entirety in definitions~\ref{def:s}~and~\ref{def:c}.
\end{ex}

First we must define how to successively partition the system.   We only allow for nested sequences of partitions, so that larger-scale descriptions of the system cannot contain information that smaller-scale descriptions lack.  The way in which a system is partitioned defines a sequence of descriptions of the system, and thus different partitioning schemes can be thought of as different nested ontologies with which to create successively coarser descriptions of the system.  This formulation allows for a  general framework for describing a system at multiple scales, given the constraint that nested, successively larger-scale descriptions of a system correspond to nested subsets of the system that are decreasing in size.

\subsection{Definition}
We now formally define this class of complexity profiles.  To do so, we first build up some notation for defining nested sequences of partitions:

\begin{defn}
\label{def:nps}
Define $P=\{P_i\}_{i=1}^\infty$ to be a \emph{nested partition sequence} of a set $X$ if each $P_i$ is a partition of $X$, $P_i\leq P_j$ (i.e. $P_i$ is a refinement of $P_j$) whenever $i>j$, and $P_i<P_j$ (i.e. $P_i$ is a strict refinement of $P_j$) whenever $|X|\geq i>j$.  
\end{defn}

Note that, in order for the strict refinement clause of this definition to be satisfied (i.e. for $P_i$ to have more parts than $P_j$ whenever $|X|\geq i>j$), it must be that $P_n$ contains $n$ parts for $n\leq |X|$ and $P_n=P_{|X|}$ for $n>|X|$, since a partition of $X$ cannot have more than $|X|$ parts.  


\begin{ex}
\label{ex:4p}
Let $X=\{x_1,x_2,x_3,x_4\}$.  An example of a nested partition sequence of $X$ is $P=(P_1,P_2,P_3,P_4,P_5,...)=(\{\{x_1,x_2,x_3,x_4\}\},~~\{\{x_1,x_3\},\{x_2,x_4\}\},~~\{\{x_1,x_3\},\{x_2\},\{x_4\}\},~~\{\{x_1\},\{x_2\},\{x_3\},\{x_4\}\},~~\{\{x_1\},\{x_2\},\{x_3\},\{x_4\}\},...)$ 
\end{ex}

\begin{defn}
\label{def:s}
Given a nested partition sequence $P$ of a system $X$, \newline
define $\tilde{S}_X^P(n)\equiv nS^P_X(n)\equiv \sum_{\chi \in P_n} H(\chi)$ for $n\in\mathbb{Z}^+$. 
\end{defn}

Note that $\tilde{S}^P_X(n)$ is non-decreasing in $n$ and captures the total (potentially overlapping) information of the system parts, while $S^P_X(n)$ is the average amount of information necessary to describe one of the $n$ parts.  Information that is $n$-fold redundant (i.e. is of scale $n$) can be counted up to $n$ times in $\tilde{S}^P_X(n)$---it is this fact that motivates the following definition of a complexity profile. 

\begin{defn}
\label{def:c}
Given a nested partition sequence $P$ of a system $X$, the complexity profile $C^P_X(n):\mathbb{Z}^+\rightarrow [0,\infty)$ is defined as $C^P_X(n)=\tilde{S}^P_X(n)-\tilde{S}^P_X(n-1)$, with the convention that $\tilde{S}^P_X(0)=0$.
\end{defn}

\begin{remark}
\label{lem:C}
For $n=1$, $C^P_X(n)=H(X)$.  For $n>|X|$, $C^P_X(n)=0$.  And for $1<n\leq |X|$, $C^P_X(n)=H(A)+H(B)-H(A,B)= I(A;B)$ where $A$ and $B$ are the two subsets of $X$ that are elements of $P_n$ but not of $P_{n-1}$, and where $I$ denotes mutual information.  Thus, this complexity profile is very computationally tractable. 
\end{remark}

\begin{ex}
Using the nested partition sequence given in \cref{ex:4p} of $X=\{x_1,...,x_4\}$, if $x_i$ are unbiased bits, $x_1=x_2$, and $x_1,x_3,x_4$ are mutually independent, we have $\tilde{S}_X^P(1)=H(X)=3$, $\tilde{S}_X^P(2)=H(x_1,x_3)+H(x_2,x_4)=4$, and $\tilde{S}_X^P(n)=4$ for $n>2$.  Thus $C^P_X(1)=3$, $C^P_X(2)=1$, and $C^P_X(n)=0$ for $n>2$.
\end{ex}

\begin{ex}
Consider a system $X$ of $N$ molecules, the velocities of which are independently drawn from a Maxwell-Boltzmann distribution for which the temperature $T$ is itself a random variable.  Consider a nested partitioning scheme in which at each step the largest remaining part (or, in the case of a tie, one of the largest remaining parts) is divided as equally as possible in two.  The resulting complexity profile will then have $C(1)=H(X)$ and $C(n)=H(T)$ for $1<n<<N$, since for $n<<N$, the size of the parts will be large enough so that $T$ can be almost precisely determined from any single part.  As $n$ approaches $N$, a measurement of any single part will yield more and more uncertainty regarding the value of $T$ and so $C(n)$ will slowly decay from $H(T)$ to $0$.  Such a complexity profile captures the fact that at the smallest scale, there is a lot of information related to the microscopic details of each molecule, but at a wide range of larger intermediate scales, the information present is much smaller and roughly constant, arising only from the common large-scale influence that temperature has across the system.
\end{ex}

\subsection{The multi-scale law of requisite variety and the sum rule}

This complexity profile roughly captures the notion of redundancy and will satisfy the properties described in definitions~\ref{def:mono} and~\ref{def:sum} (as proved below).  It is dependent on the particular set of partitions used---a reflection of the fact that there are multiple ways to coarse-grain a system---and thus will not capture the redundancies present in an absolute sense, as the complexity profile described in refs.~\cite{original, allen2017multiscale} does.  But that complexity profile, while it does obey the sum rule, does not manifest the multi-scale law of requisite variety.  Thus, while it characterizes the information structure present in a system, it does not allow us to compare a system to its environment in a mathematically rigorous way.  The class of complexity profiles considered here allows this comparison by requiring that the system and the environment be partitioned in the same way, breaking permutation symmetry and accounting for the correspondence between system and environmental components.

\begin{thm}
\label{thm:match}
Multi-scale law of requisite variety.  If a system $X$ matches its environment $(Y,f)$, then for all nested partition sequences $P$ of $X$, $C^P_X(n)\geq C^{P^f}_Y(n)$ at each scale $n$, where $P^f$ is the corresponding nested partition sequence of $Y$ (see definition~\ref{def:fp} below).
\end{thm}

\begin{proof}
See \cref{sec:proofs}.
\end{proof}

\begin{defn}
\label{def:fp}
Given a nested partition sequence $P$ of a set $X$ and a bijection $f:Y\rightarrow X$, define $P^f$ to be the nested partition sequence of $Y$ such that $\forall n\in\mathbb{Z}^+$, $y_1,y_2\in Y$ belong to the same part of $P^f_n$ iff $f(y_1),f(y_2)\in X$ belong to the same part of $P_n$.
\end{defn}

One of the advantages of having the complexity profile depend on the partitioning scheme is that theorem~\ref{thm:match} holds for all possible nested partition sequences of the system, assuming its environment is partitioned in the same way.  In other words, regardless of how the partitions are used to define scale, the system must have at least as much complexity as its environment at all scales, so long as scale is defined in the same way for the system and its environment. 

Furthermore, not only must all possible complexity profiles of the system match the corresponding complexity profile of the environment, but all possible complexity profiles of all possible subsets of the system must match the corresponding complexity profile of the corresponding subset of the environment, as stated in the following corollary to theorem~\ref{thm:match}.  This is a powerful statement, since it implies that not only must the system have at least as much complexity as its environment at all scales, but also that subdivisions within the system must be aligned with the corresponding subdivisions within the environment (see section 2.6 of ref.~\cite{siegenfeld2020}).

\begin{cor}
\label{cor:sm}
Subdivision matching.  Suppose a system $X$ matches its environment $(Y,f)$. Then for any subsets $Y'\subset Y$ and $X'=f(Y')\subset X$, $C^P_{X'}(n)\geq C^{P^f}_{Y'}(n)$ at each scale $n$ for all nested partition sequences $P$ of $X'$.  
\end{cor}

\begin{proof}
Since $X$ matches $Y$, $X'$ matches $Y'$.  Therefore theorem~\ref{thm:match} applies to $X'$ and $Y'$.
\end{proof}

We now state and prove the sum rule:

\begin{thm}
\label{thm:sum}
Sum rule.  For any system $X$ and all nested partition sequences $P$ of $X$, $\sum_{n=1}^\infty C^P_X(n)=\sum_{x\in X} H(x)$
\end{thm}

\begin{proof}
$\sum_{n=1}^\infty C^P_X(n)=\lim_{n\rightarrow\infty} \tilde{S}^P_X(n)-\tilde{S}^P_X(0)=\tilde{S}^P_X(|X|)-0=\sum_{x\in X} H(x)$
\end{proof}

Because the complexity at each scale measures the amount of additional information present when different parts of the system are considered separately, the sum of complexity across all scales will simply equal the total information present in the system when each component is considered independently of the rest.  Thus, given fixed individual behaviors of the system components, there is a necessary tradeoff between complexity at larger and smaller scales regardless of which partitioning scheme is used (i.e. regardless of how scale is defined). 

\subsection{Choosing from among the partitioning schemes}
Because of the dependence on the partitioning scheme, \cref{def:c} defines a family of complexity profiles.  That there is no single complexity profile for this definition can be thought of as a consequence of their being no single way to coarse-grain a system.  In other words, implicit in any particular complexity profile of a system is a scheme for describing that system at multiple scales.  While there is no such scheme that is ``the correct scheme'' in an absolute sense, for any particular purpose (and often for almost any conceivable purpose), some schemes are far better than others.  

But before examining this question, we first consider a strong advantage of the multiplicity of complexity profiles: \cref{thm:match} applies to all of them.  Thus, any complexity profile, regardless of the partitioning scheme, can potentially be used to show a multi-scale complexity mismatch between system and environment.  This is useful when one has information about the probability distributions of the system and environment separately but not necessarily on the joint probability distribution of system and environment together, such that one cannot directly determine whether the system matches the environment (since quantities such as $H(y|x)$ would be unknown for any given system component $x$ and environmental component $y$).

Assuming one knows which system components correspond to which environmental components, one can test for potential incompatibility between the system and environment by considering \textit{any} nested partition sequence of any subset of the system and the corresponding subset of the environment, as per theorem~\ref{thm:match} and corollary~\ref{cor:sm}.  Thus, a meaningful comparison of system and environment can be made for a wide variety of complexity profiles, provided the definitions for the system and environment are consistent.  In the likely case that the complexity profiles cannot be precisely calculated, this framework thus supports a wide variety of qualitative complexity profiles that one may wish to construct.  

When the correspondence between system and environmental components is unknown, there are still ways in which to compare the system and the environment.  For instance, if a system $X$ matches its environment $Y$, then 
\begin{equation}
\label{eq:max}
\max_P F(C^P_X)\geq \max_P F(C^P_Y)
\end{equation}
for all functions $F$ that map complexity profiles onto $\mathbb{R}$ and are non-decreasing in $C(n)$ for all scales $n$.  Thus, finding even a single function $F$ for which \cref{eq:max} does not hold is enough to show that $X$ cannot possibly match $Y$, regardless of how they may be connected.  Other such constructions that are independent of the bijection between $X$ and $Y$ are also possible.

However, although any partitioning scheme can be used to show a mismatch between a system and its environment, not all partitioning schemes are equally good choices for gaining an understanding of the structure of the system.  Each part of a partition represents approximating the system by describing only that subset of its components, and so, if the purpose of the complexity profile is to characterize the structure of the system, the partitions should be chosen accordingly.  For instance, for a system $\{x_1,x_2,x_3,x_4\}$ where $x_1=x_2$ and $x_3=x_4$, partitioning the system into $\{x_1,x_2\}$ and $\{x_3,x_4\}$ does not make sense if the goal is to create a reasonably faithful two-component description of the four-component system.  

As a heuristic, successive cuts in a nested partition sequence should cut through random variables with significant mutual information (i.e. significant redundancy), although, of course, taking a greedy algorithm (i.e.. first maximizing complexity at scale $2$, and then choosing the next partition to maximize complexity at scale $3$, given the constraint that it has to be nested within the previous partition, and so on) may not always match system structure.  Nonetheless, this greedy algorithm does at least provide a consistent way to define complexity profiles across various systems such that complexity is decreasing with scale.\footnote{Formally, we can define this complexity profile using the nested partition sequence $P$ that maximizes the complexity profile according to a ``dictionary ordering'' in which $C^{P_1}_X>C^{P_2}_X$ if there exists an $n$ such that $C^{P_1}_X(n)>C^{P_2}_X(n)$ and for all $m<n$, $C^{P_1}_X(m)=C^{P_2}_X(m)$.  Equivalently, this $P$ maximizes $\sum_{n=1}^\infty M^{-n}C_X^P(n)$ for any $M>C_X^P(1)=H(X)$.  However, just because $P$ maximizes the complexity profile for $X$ according to this (or any other) metric does not guarantee that for an environment $(Y,f)$ of $X$, $f(P)$ will maximize the complexity profile of $Y$ according to the same metric.}


\begin{ex}
\label{ex:blocks}
Consider a system $X=\{x_1,x_2,...,x_8\}$ such that $x_1=x_2=x_3$, $x_4=x_5=x_6$, and $x_7=x_8$, but otherwise all components are mutually independent (i.e. $x_1,x_4,x_7$) are all mutually independent (\cref{fig:blocks}).  Then intuitively we expect that $C_X(1)=C_X(2)=H(x_1)+H(x_4)+H(x_7)$, $C_X(3)=H(x_1)+H(x_4)$, and $C_X(n)=0$ for $n>3$.  A nested partition sequence that gives us this complexity profile is $P=(P_1,P_2,P_3,...)$ with $P_1=\{X\}$, $P_2=\{\{x_1,x_4,x_7\},\{x_2, x_3, x_5, x_6, x_8\}\}$, $P_3=\{\{x_1,x_4,x_7\},\{x_2, x_5, x_8\},\{x_3, x_6\}\}$, and where it does not matter which subsequent partitions are used, since each part of $P_3$ contains mutually independent random variables.
\end{ex}

\begin{figure}
\begin{center}
\begin{subfigure}[t]{0.45\textwidth}
\centering
\includegraphics[width=.8\linewidth]{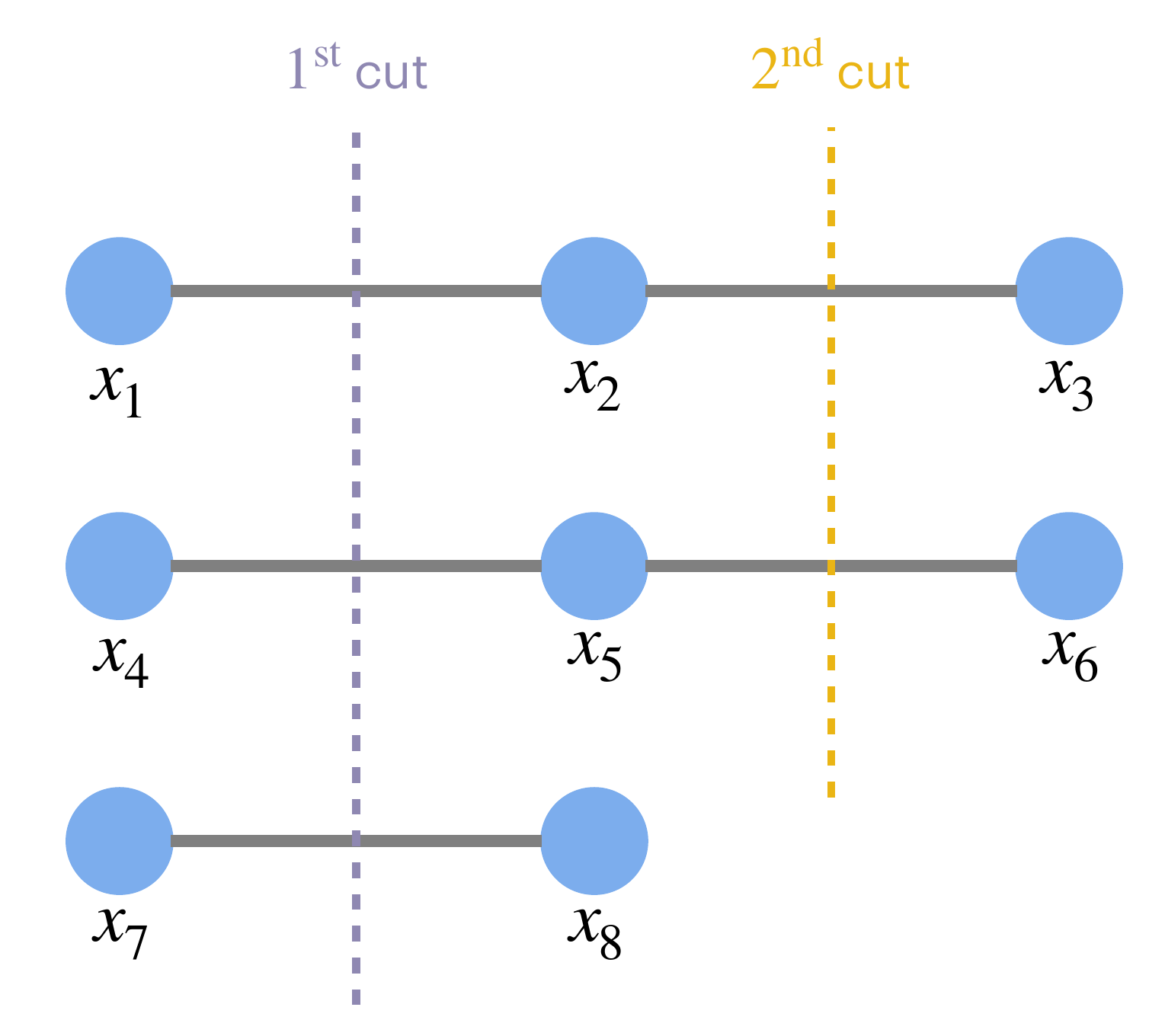}
\caption{}
\label{sf:b}
\end{subfigure}
\hfill
\begin{subfigure}[t]{0.45\textwidth}
\centering
\includegraphics[width=\linewidth]{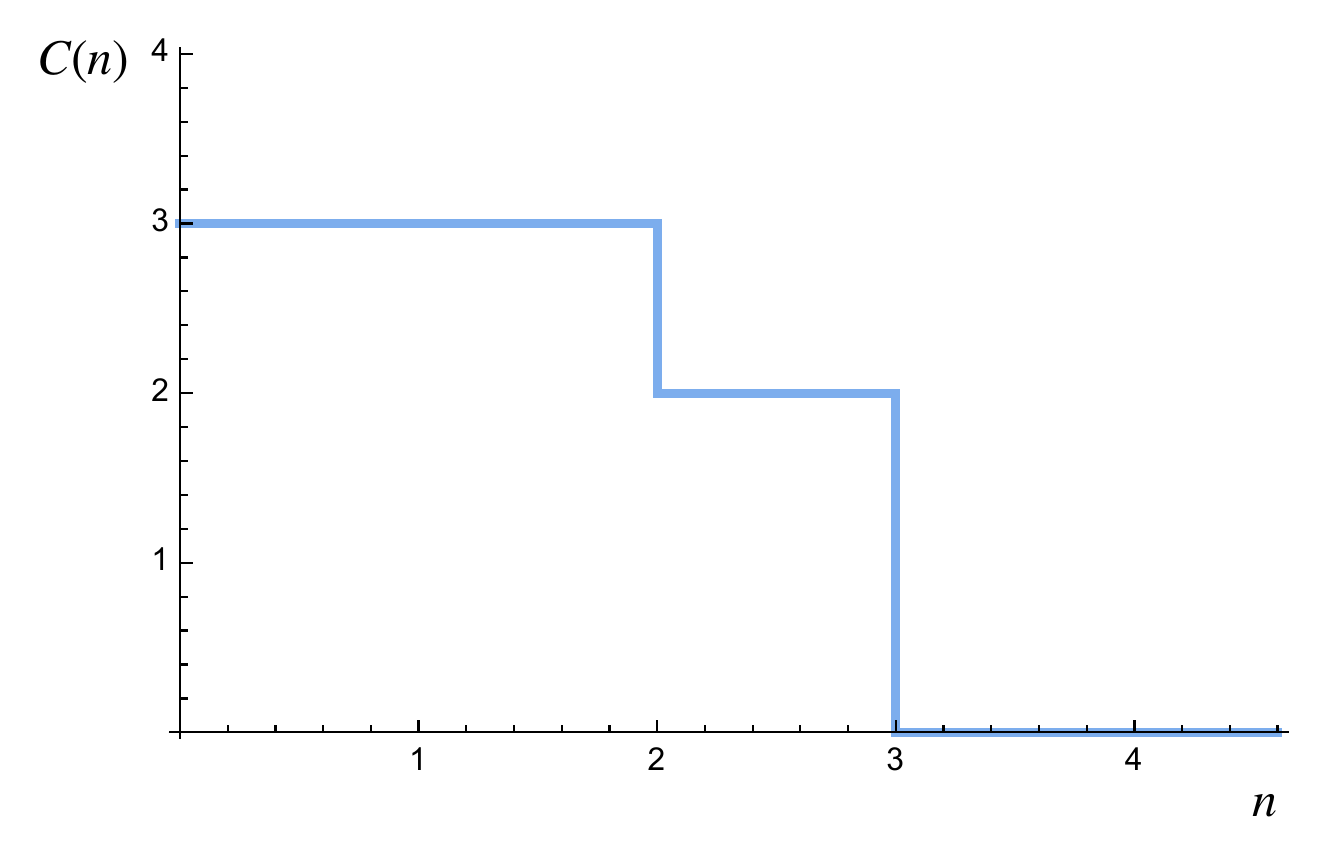}
\caption{}
\label{sf:bcp}
\end{subfigure}
\caption{(\subref{sf:b}) The first and second cuts necessary to create the first three partitions discussed in \cref{ex:blocks} are shown, together with (\subref{sf:bcp}) the resulting complexity profile (made continuous via \cref{eq:ctilde}) if $H(x)=1$ for each $x\in X$.}
\label{fig:blocks}
\end{center}
\end{figure}

\begin{ex}
\label{ex:blanket}
Consider a two dimensional 4x4 grid of random variables in which two variables have nonzero mutual information conditioning on the rest if and only if they are adjacent.  One way to partition the grid that respects its structure is given in \cref{fig:blanket}.
\end{ex}

\begin{figure}
\begin{center}
\begin{subfigure}[t]{0.43\textwidth}
\centering
\includegraphics[width=\linewidth]{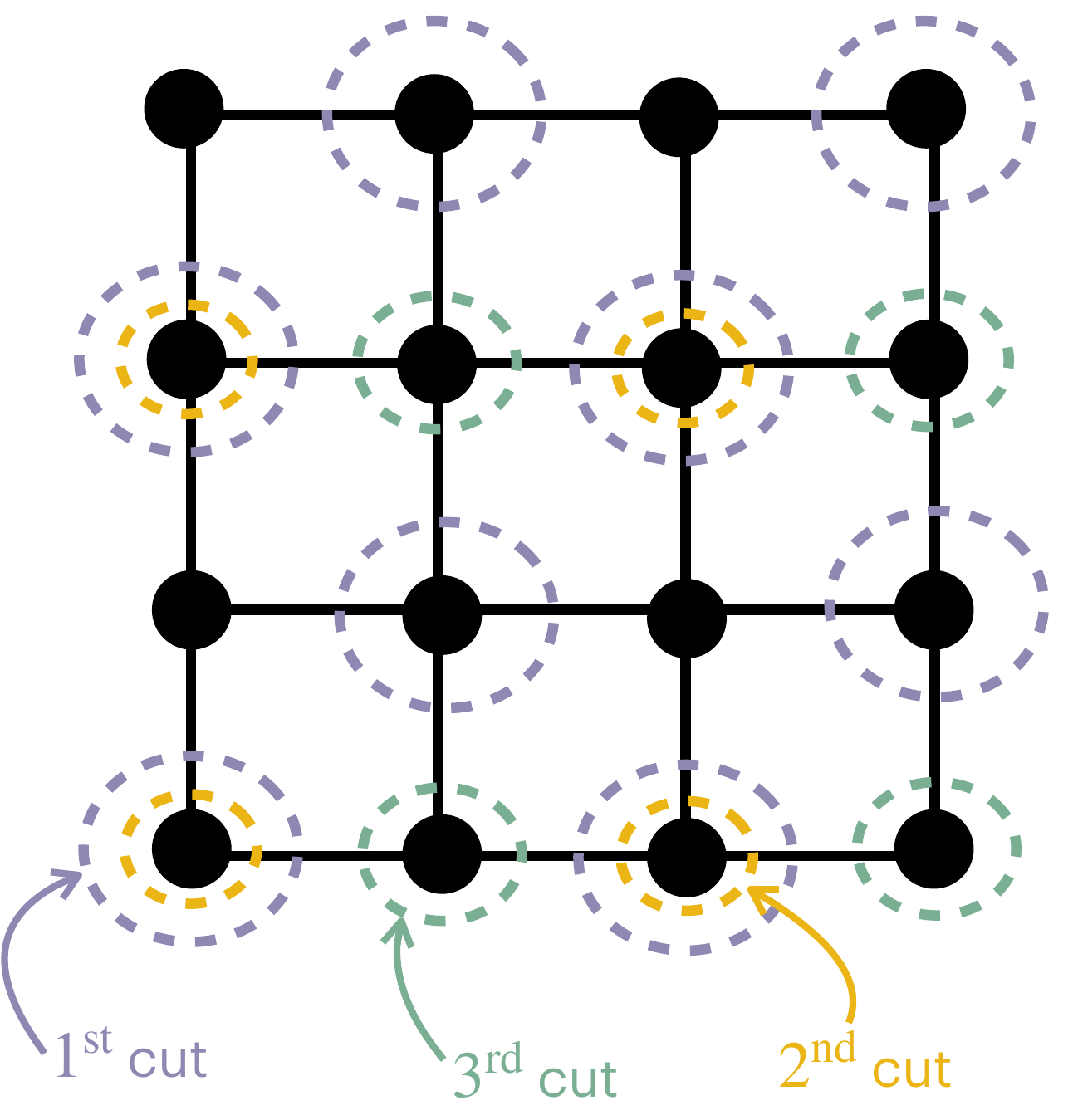}
\caption{}
\label{sf:l}
\end{subfigure}
\hfill
\begin{subfigure}[t]{0.49\textwidth}
\centering
\includegraphics[width=\linewidth]{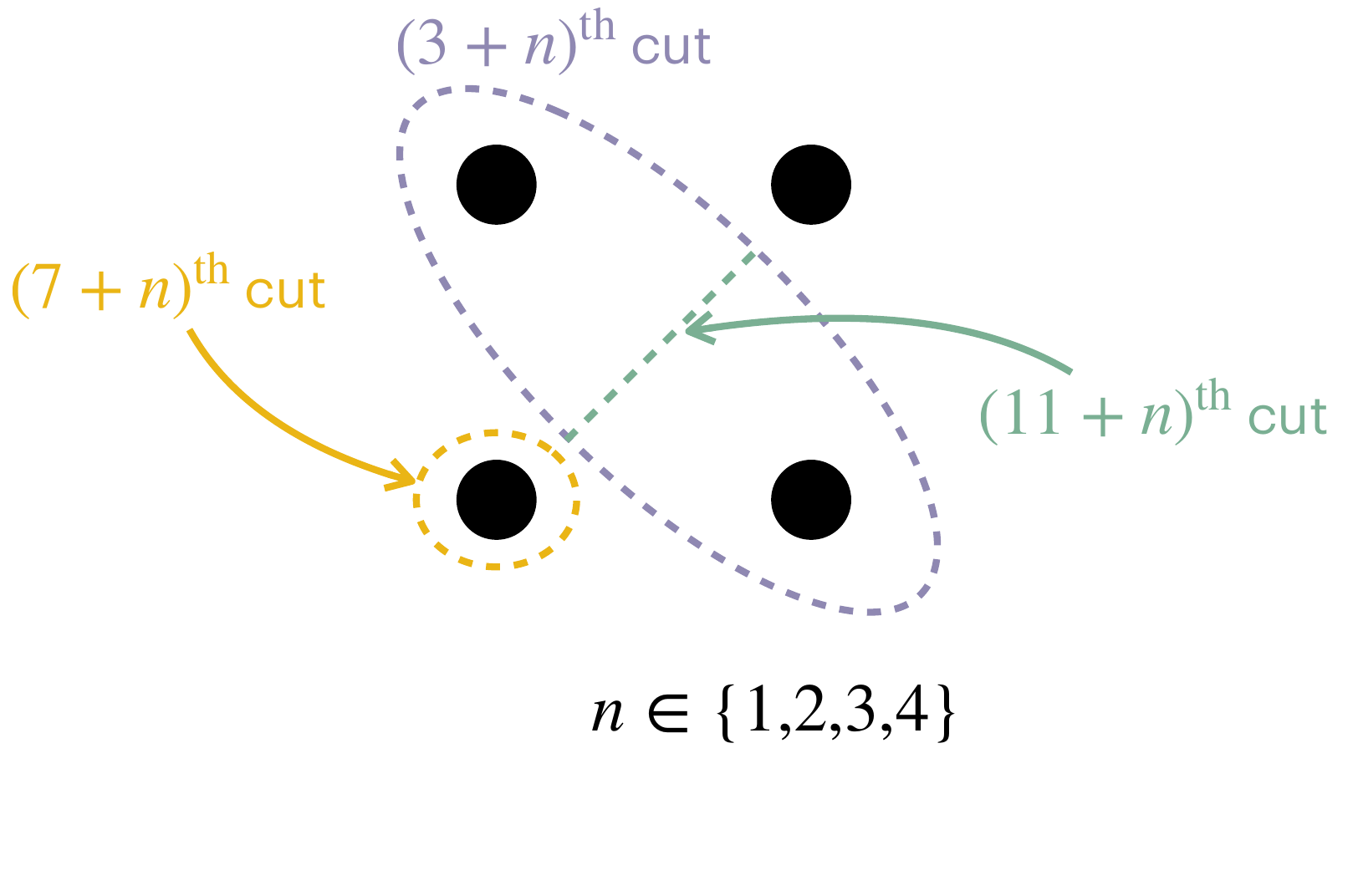}
\caption{}
\label{sf:l4}
\end{subfigure}
\caption{Cuts that will create a nested partition sequence for the 4x4 grid of random variables described in \cref{ex:blanket} are shown.  The first three cuts partition the grid into four parts with four components each, as shown in (\subref{sf:l}).  Each of the four parts (labeled by $n\in\{1,2,3,4\}$) are then subsequently partitioned as shown in (\subref{sf:l4})  .  The resulting complexity profile will of course depend on the nature of the random variables and their correlations.}
\label{fig:blanket}
\end{center}
\end{figure}

\begin{ex}
\label{ex:h}
Consider a hierarchy consisting of 7 individuals: a leader with two subordinates, each of which have two subordinates themselves, as depicted in \cref{fig:h}.  The behavior of each individual is represented by 3 random variables, each with complexity $c$.  Thus if examined separately from the rest of the system, the complexity of each individual is $3c$.  On the left, everyone completely follows the leader, resulting in a complexity of $3c$ up to scale 7 regardless of the partitioning scheme.  On the right, some information is transmitted down the hierarchy but lower levels are also given some autonomy, resulting in more complexity at smaller scales but less at larger scales (but with the same area under the curve, consistent with the sum rule).
\end{ex}

\begin{figure}
\begin{center}
\begin{subfigure}[t]{0.45\textwidth}
\centering
\includegraphics[width=\linewidth]{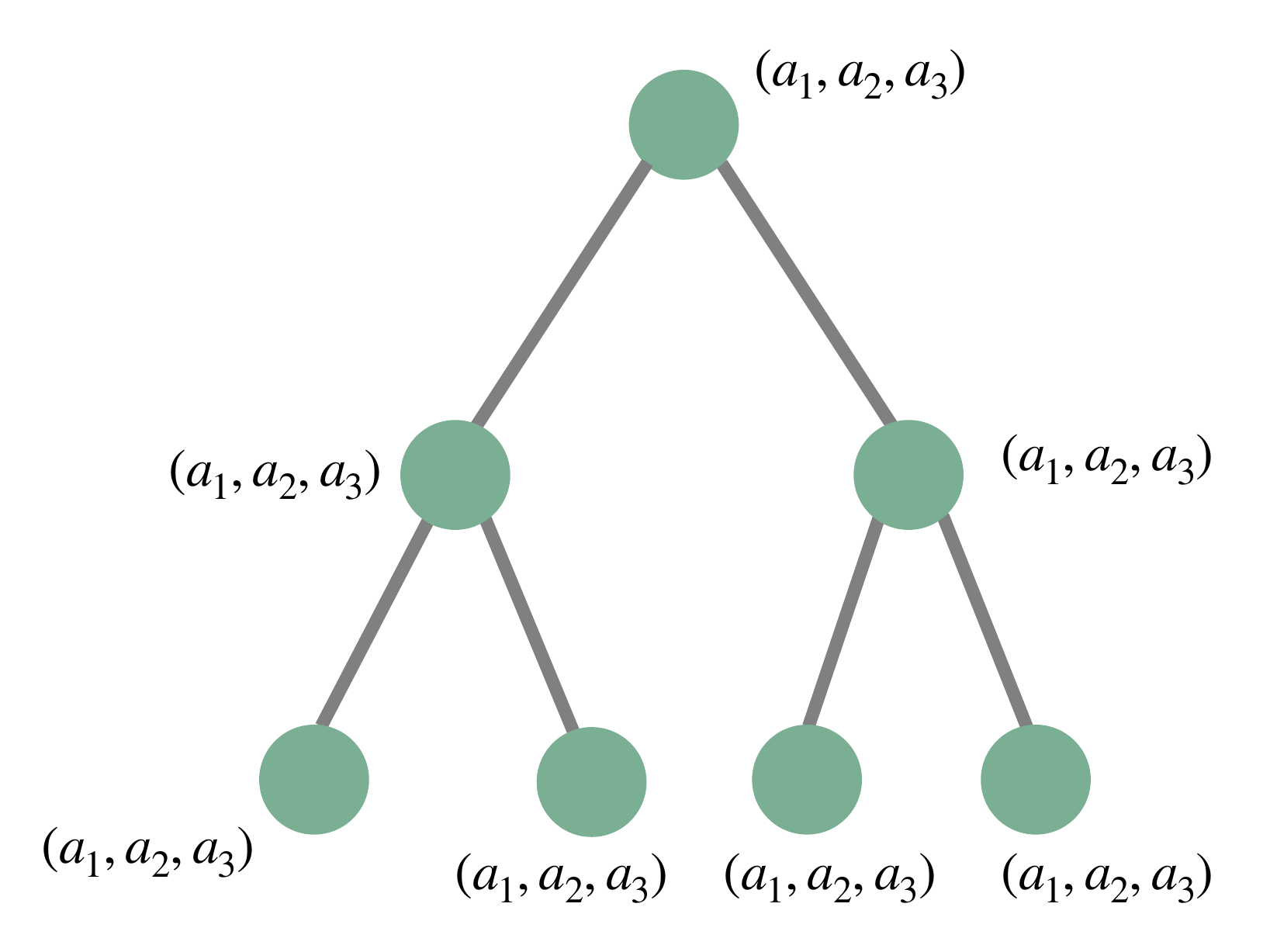}
\caption{}
\label{sf:ha}
\end{subfigure}
\hfill
\begin{subfigure}[t]{0.45\textwidth}
\centering
\includegraphics[width=\linewidth]{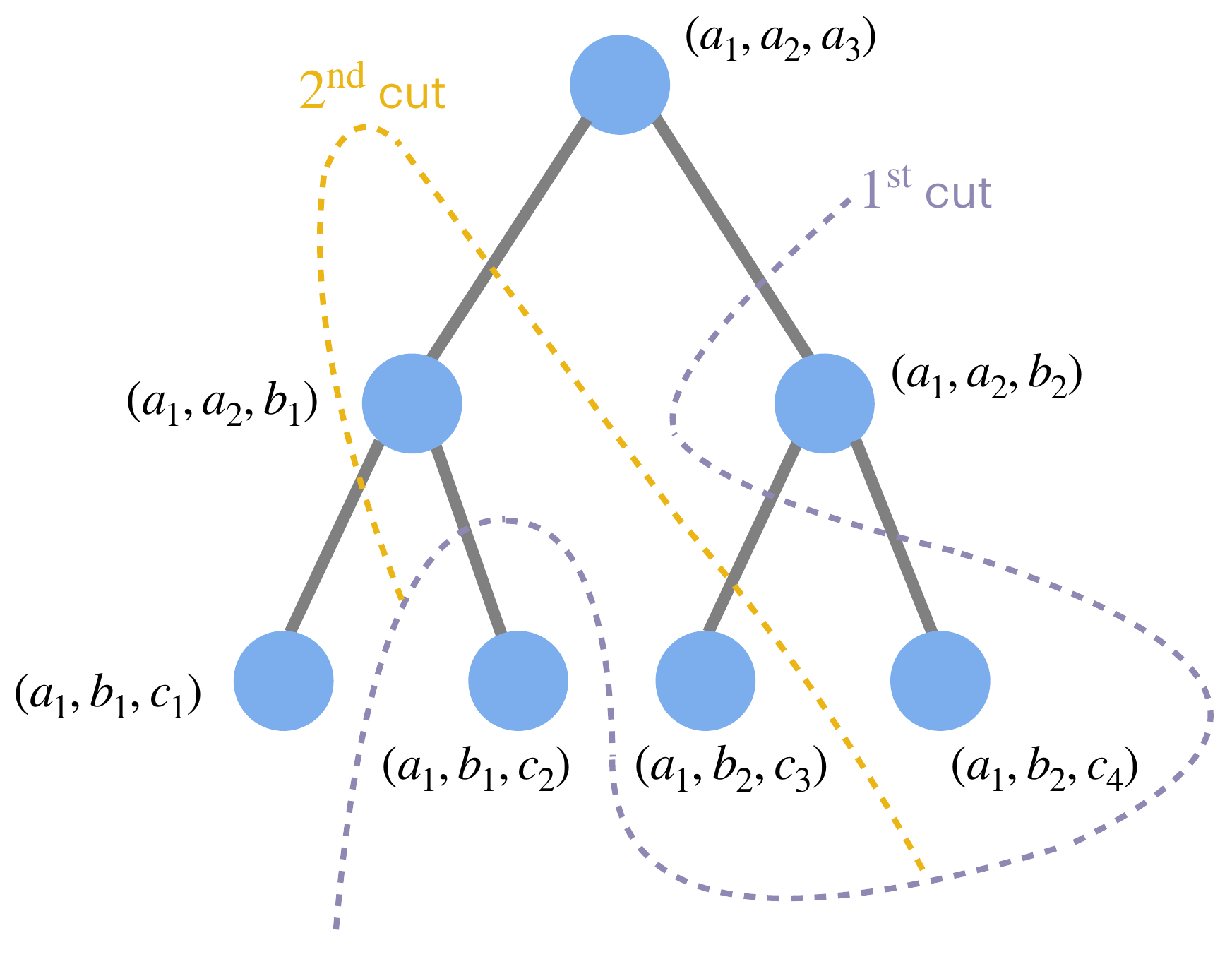}
\caption{}
\label{sf:hb}
\end{subfigure}

\vspace{1em}
\begin{subfigure}[t]{0.45\textwidth}
\centering
\includegraphics[width=\linewidth]{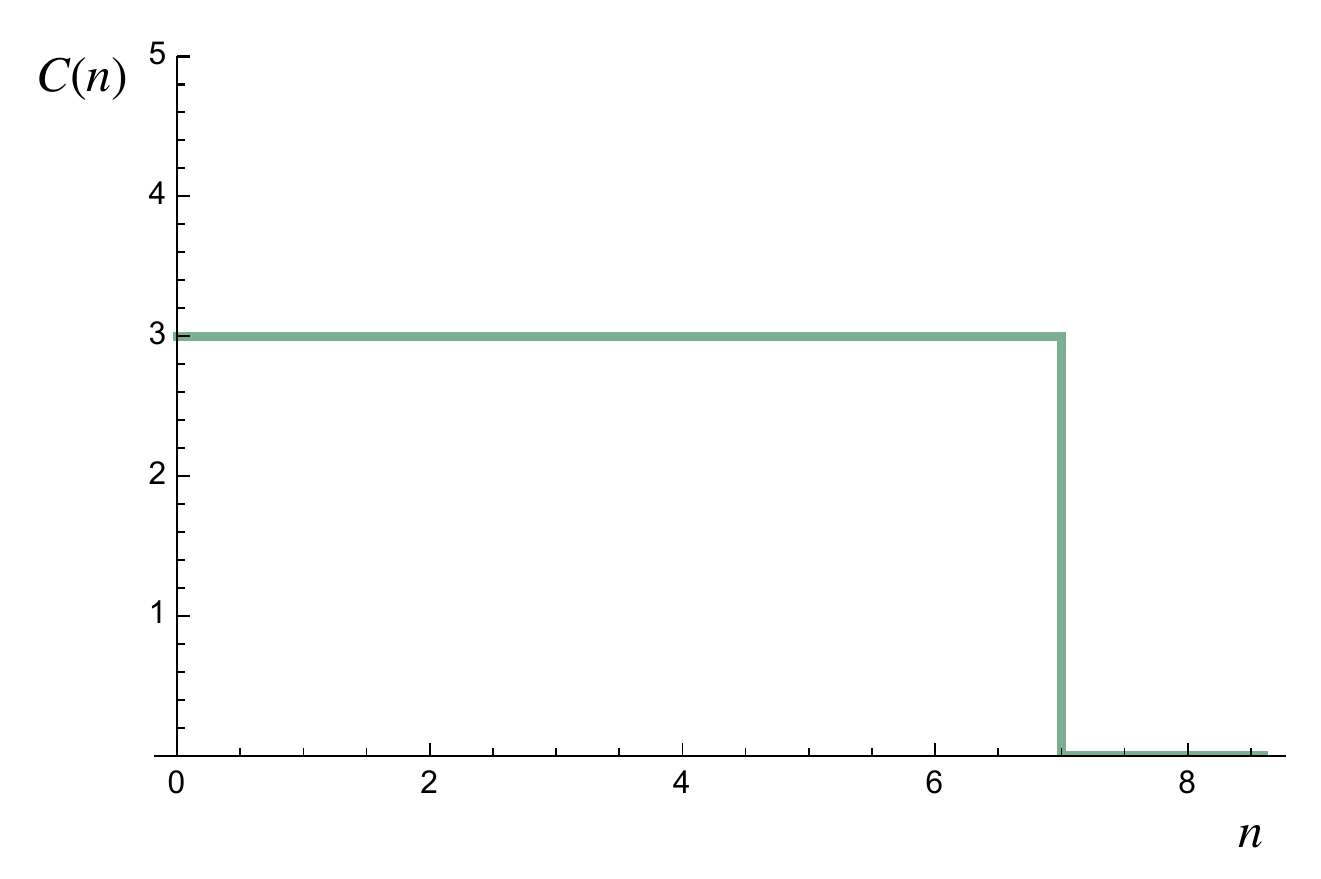}
\caption{}
\label{sf:hc}
\end{subfigure}
\hfill
\begin{subfigure}[t]{0.45\textwidth}
\centering
\includegraphics[width=\linewidth]{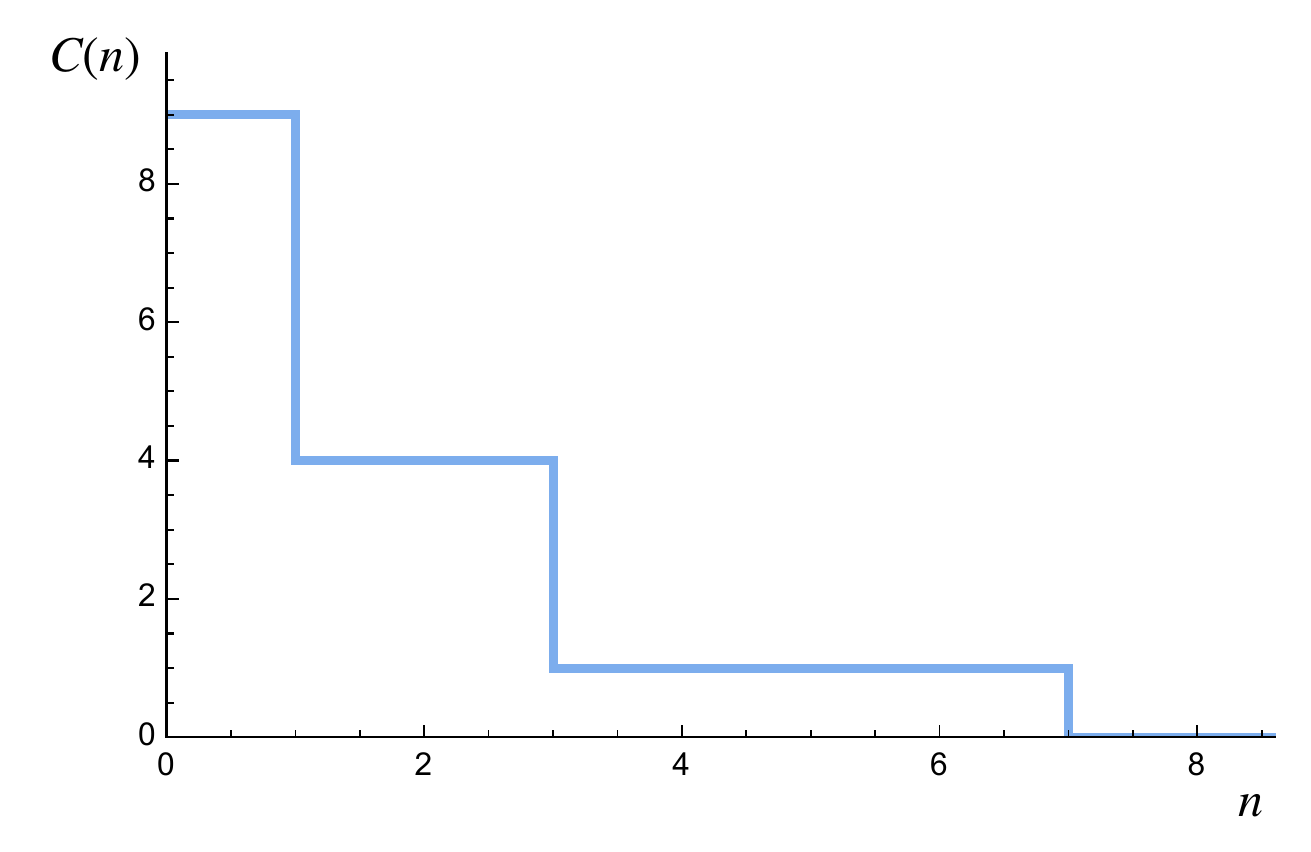}
\caption{}
\label{sf:hd}
\end{subfigure}
\caption{The hierarchies discussed in \cref{ex:h} are shown, together with the resulting complexity profiles (made continuous via \cref{eq:ctilde}) for $c=1$.  The complexity profile in (\subref{sf:hc}) can be obtained from any nested partition sequence for the hierarchy in (\subref{sf:ha}).  The complexity profile in (\subref{sf:hd}) can be obtained from any nested partition sequence for which the first three partitions are given by the displayed cuts in the hierarchy in (\subref{sf:hb}). }
\label{fig:h}
\end{center}
\end{figure} 

\subsection{Combining subsystems}
If a system can be divided into independent subsystems, then the complexity profile of the system as a whole can be written as the sum of complexity profiles of each of the independent subsystems.   And if a system can be divided into $m$ subsystems that behave identically, its complexity profile will equal that of any one of the subsystems except with the scale axis stretched by a factor of $m$.  These properties are made precise below.



\begin{thm}
\label{thm:indep}
Additivity of the complexity profiles of superimposed independent systems.  Suppose two disjoint systems $A$ and $B$ are independent, i.e. the mutual information $I(A;B)=0$, and let $C=A\cup B$.    Consider any nested partition sequences $P^A$ of $A$ and $P^B$ of $B$.  Then for all nested partition sequences $P^C$ that restrict to $P^A$ on $A\subset C$ and to $P^B$ on $B\subset C$,  
\begin{equation}
C^{P^C}_C(n)=C^{P^A}_A(n)+C^{P^B}_B(n)
\end{equation}
In other words, the complexity profiles of independent subsystems add.
\end{thm}

\begin{proof}
This result follows from $P^C_i$ restricting to $P^A_i$ on $A$ and $P^B_i$ on $B$ and the fact that for any subsets $A'\subset A$, $B'\subset B$, $H(A'\cup B')=H(A')+H(B')$. 
\end{proof}

In order to formulate the second property, we first build up some notation in the following two definitions:

\begin{defn}
\label{def:copy}
For a system $X$ and positive integer $m$, let $m*X=\cup_{i=1}^mX_i$ where the $X_i$ are disjoint systems for which there exist bijections $f_i:X_i\rightarrow X$ such that $\forall x\in X_i$, $H(x|f_i(x))=H(f_i(x)|x)=0$.  In other words, $m*X$ contains $m$ identical copies of $X$, such that the behavior of any one copy completely determines the behavior of all of the others.
\end{defn}

\begin{defn}
\label{def:mp}
Given a nested partition sequence $P$ of a system $X$ and a positive integer $m$, define the nested partition sequence $m*P$ of the system $m*X=\cup_{i=1}^mX_i$ (with the bijections $f_i:X_i\rightarrow X$) as follows. For $n\leq m$, $(m*P)_n\equiv\{\cup_{i=n}^mX_i\} \cup \{X_i: i<n\}=\{X_1,X_2,...,X_{n-1},\cup_{i=n}^mX_i\}$.  For $n\geq m$, define $(m*P)_n$ such that it restricts to $P^{f_i}_{n_i}$ (see \cref{def:fp}) on each $X_i\subset m*X$, where $n_i=\lceil n/m\rceil$ if $i\leq (n \mod m)$ and $n_i=\lfloor n/m \rfloor$ otherwise.
\end{defn}

\begin{thm}
\label{thm:dep}
Scale-additivity of replicated systems.  Let $P$ be any nested partition sequence of a system $X$.  Then 
\begin{equation}
C^{m*P}_{m*X}(n)=C^P_X(\lceil n/m\rceil)
\end{equation}
In other words, the effect of including $m$ exact replicas of $X$ is to stretch the scale axis of the complexity profile by a factor of $m$.
\end{thm}

\begin{proof}
This result follows from definitions~\ref{def:copy}~and~\ref{def:mp}.
\end{proof}

Theorems~\ref{thm:indep} and~\ref{thm:dep} indicate that for any block-independent system $X$, there exists a nested partition sequence that yields the same complexity profile as that given by the formalism in refs.~\cite{allen2017multiscale}~and~\cite{original}.\footnote{The formalism in ref.~\cite{allen2017multiscale}/ref.~\cite{original} is stated in ref.~\cite{original} to be the only such formalism that is a linear combination of entropies of subsets of the system, that yields its results for block-independent systems, and that is symmetric with respect to permutations of the components.  The partition formalism in this paper does not contradict this statement, since any partitioning scheme will break this permutation symmetry for systems with greater than two components.}

\section{Conclusion}
\label{sec:conclusion}
The motivation behind our analysis has been to construct a definition of a complexity profile for multi-component systems that obeys both the sum rule and a multi-scale version of the law of requisite variety.  In order to do so, we first had to generalize the law of requisite variety to multi-component systems.  We then created a formal definition for a complexity profile and defined two properties---the multi-scale law of requisite variety and the sum rule---that complexity profiles should satisfy.  Finally, we constructed a class of examples of complexity profiles and proved that they satisfy these properties.  We demonstrate their application to a few simple systems and show how they behave when independent and dependent subsystems are combined.  

This formalism is purely descriptive, in that questions of causal influence and mechanism (i.e. what determines the states of each component) are not considered; rather only the possible states of the system and its environment and correlations among these states are considered.   
(Cf. statistical physics not considering the Newtonian dynamics of individual gas molecules but rather only the probabilities of finding the gas in any given state.)  
By abstracting out notions of causality and mechanism, this approach allows for an understanding of the space of all possible system behaviors and for an identification of systems that are doomed to failure regardless of mechanism.   
The way in which system and environmental components are mechanistically linked and the evolution and adaptability of complex systems over time are directions for future work.

Complexity profiles other than those presented \cref{sec:class} may exist, including profiles that take advantage of some known structure of the systems under consideration; for instance, for systems that can be embedded into $\mathbb{R}^d$ where $d$ is far lower than the number of system components, Fourier methods could be explored.  More broadly, the sum rule could be relaxed: while completely eliminating any tradeoff of complexity among scales would likely lead to under-constrained profiles---certainly, smaller-scale complexity must be reduced in order to create larger-scale structure---formalizations other than $\sum_n C_X(n)=\sum_{x\in X}H(x)$ could be considered.  One could also imagine other definitions of what it means for a system to match its environment.  But just as the sum rule could be modified but the tradeoff of complexity among scales should nonetheless be manifest, some sort of multi-scale law of requisite variety is necessary if the complexity profiles of multiple systems are to be meaningfully compared.

At the very least, the multi-scale law of requisite variety and sum-rule-constrained profiles provide a formal grounding that supports the conceptual implications of scale-dependent complexity.  Our hope is that these formalisms spur further development in our understanding of the general properties of multi-component systems. 

\subsection*{Acknowledgements}
This material is based upon work supported by the National Science Foundation Graduate Research Fellowship under Grant No. 1122374, the Hertz Foundation, and the Long-Term Future Fund.  We would also like to thank Alex Zhu and Robi Bhattacharjee for helpful discussions. 

\beginsupplement

\section*{Appendix}

\section{Proofs}
\label{sec:proofs}
In order to prove \cref{thm:match}, we first prove the following lemma.  

\begin{lemma} 
\label{lem:i}
If $H(B_1|A_1)=H(B_2|A_2)=0$, then $I(A_1;A_2)\geq I(B_1;B_2)$.
\end{lemma}
\begin{proof}
\begin{align*}
&I(A_1;A_2)=H(A_1)+H(A_2)-H(A_1,A_2)=\\
&H(B_1)+H(A_1|B_1)-H(B_1|A_1)+H(B_2)+H(A_2|B_2)-H(B_2|A_2)\\
&-H(B_1,B_2) -H(A_1,A_2|B_1,B_2)+H(B_1,B_2|A_1,A_2)=\\
&I(B_1;B_2)+H(A_1|B_1)+H(A_2|B_2)-H(A_1,A_2|B_1,B_2)=\\
&I(B_1;B_2)+I(A_1;A_2|B_1;B_2)+I(A_1;B_2|B_1)+I(A_2;B_1|B_2)\geq I(B_1;B_2)
\end{align*}  
\end{proof}

We now prove \textbf{Theorem~\ref{thm:match}}:
\textit{
Multi-scale law of requisite variety.  If a system $X$ matches its environment $(Y,f)$, then for all nested partition sequences $P$ of $X$, $C^P_X(n)\geq C^{P^f}_Y(n)$ at each scale $n$.}\newline

\begin{proof}
For any collections of random variables $A=\{a_1,..,a_N\}$ and $B=\{b_1,...,b_N\}$, $\forall i~H(b_i|a_i)=0$  implies that $0\leq H(B|A)\leq\sum_i H(b_i|A)\leq\sum_i H(b_i|a_i)=0$ and thus $H(B|A)=0$.  Using this fact together with remark~\ref{lem:C} and lemma~\ref{lem:i}, we obtain $C^P_X(n)\geq C^{P^f}_Y(n)$ for $n\in\mathbb{Z}^+$.  Note that $X$ and $Y$ being partitioned in the same way guarantees that for a subset $A\subset X$ and the corresponding subset $B\subset Y$, $H(B|A)=0$.
\end{proof}

\section{Continuum limit}
\label{sec:continuum}
Complexity profiles can also be defined for continuous systems.  

\begin{defn}
\label{def:cont}
We define a continuous system $X$ of size $L$ as a sequence of discrete systems $\{X_i\}_{i=1}^\infty$ with components of size $l_i\equiv L/|X_i|$ such that $X_i\subset X_j$ whenever $i<j$ and $\lim_{i\rightarrow\infty}l_i=0$. Then the complexity profile for the continuous system $X$ is defined as 
\begin{equation}
\label{eq:cont}
C_X(s)\equiv  \lim_{i\rightarrow\infty} \tilde C_{X_i}(s)= \lim_{i\rightarrow\infty}  C_{X_i}(\lceil s/l_i \rceil)
\end{equation}
provided such a limit exists, where $\tilde C_{X_i}(s)$ is defined in \cref{eq:ctilde}.
\end{defn}

\begin{remark}
\label{re:discrete}
Note that any discrete system $X$ with complexity profile $C_X(n)$ and components of size $l$ can be considered as a continuous system of size $|X|l$ per \cref{def:cont} by defining the systems $\{X_i\}_{i=1}^\infty$ (with components of size $l/i$) such that $X_i=i*X$ (\cref{def:copy}) and 
\begin{equation}
\label{eq:discrete}
\tilde C_{X_i}(s)\equiv \tilde C_X(s)=C_X(\lceil s/l\rceil)=C_{X_i}(\lceil is/l\rceil)
\end{equation}
(see \cref{eq:ctilde}).
\end{remark}

\begin{ex}
Suppose that the continuous system $X$ of size $L$ is a random continuous function $f(x)$ for $x\in [0,L]$.  Define $X_i=\{f(L/2^i),f(2L/2^i),f(3L/2^i),...,f(2^iL/2^i)\}$, so that $X_i$ has $2^i$ components, each of scale $L/2^i$.  Then $X$ can be described by the sequence $\{X_i\}_{i=1}^\infty$.
\end{ex}

We can extend the class of complexity profiles defined in \cref{sec:class} as follows.

For any nested partition sequence $P$ of a discrete system $X$ with components of size $l$, the complexity profiles in \cref{eq:discrete} can be realized by letting $C_X(n)=C^P_X(n)$ and $C_{X_i}(n)=C^{i*P}_{X_i}(n)$ (see \cref{def:mp}), since by theorem~\ref{thm:dep}, 
\begin{equation}
C^{i*P}_{X_i}(\lceil is/l\rceil)=C^P_X(\lceil s/l\rceil)
\end{equation}

To define a partition-based complexity profile using \cref{eq:cont} for a continuous system $X$ (defined by an infinite sequence of discrete systems $X_1\subset X_2\subset X_3\subset X_4...$, as per \cref{def:cont}), a nested partition sequence $P^i$ must be chosen for each $X_i$.  Of course, these nested partition sequences must be chosen so that the limit in \cref{eq:cont} exists; for consistency, it can also be required that on each $X_i\subset X_j$, each partition $P^j_n$ of $X_j$ restricts to $P^i_m$ of $X_i$ for some $m\leq n$.

\bibliography{refs}{}
\bibliographystyle{unsrt}

\end{document}